\documentclass[letterpaper, 10 pt, onecolumn,conference]{ieeeconf}  

\IEEEoverridecommandlockouts                              

\usepackage{graphicx} 
\usepackage{mathptmx} 
\usepackage{amsmath} 
\usepackage{amssymb}  
\usepackage{datetime}
\usepackage{hyperref}
\usepackage{color}

\newtheorem{lem}{Lemma}
\newtheorem{thm}{Theorem}

\def\rank{\mathrm{rank}}
\DeclareMathAlphabet{\bit}{OML}{cmm}{b}{it}

\def\<{\leqslant}           
\def\>{\geqslant}           

\def\d{\partial}
\def\wh{\widehat}
\def\wt{\widetilde}

\def\Re{\mathrm{Re} }   

\def\mR{{\mathbb R}}
\def\mC{\mathbb{C}}

\def\col{\mathrm{vec}}
\def\Tr{\mathrm{Tr}}

\def\rT{\mathrm{T}}

\def\bS{\mathbf{S}}

\def\bE{\mathbf{E}}

\def\[[[{[\![\![}
\def\]]]{]\!]\!]}

\def\bra{\langle}
\def\ket{\rangle}

\def\re{\mathrm{e}}
\def\rd{\mathrm{d}}

\def\bA{\mathbf{A}}

\def\bJ{\mathbf{J}}

\def\x{\times}
\def\ox{\otimes}

\def\fF{\mathfrak{F}}

\def\fH{\mathfrak{H}}

\def\bT{{\bf T}}

\def\mS{{\mathbb S}}

\def\eps{\epsilon}
\def\ups{\upsilon}
\def\Ups{\Upsilon}

\def\diag{\mathop\mathrm{diag}}

\begin{document}
\title{\LARGE \bf 
Decoherence Time  Maximization and Partial Isolation for Open Quantum Harmonic Oscillator Memory Networks${}^*$}

\author{Igor G. Vladimirov$^{1}$, \quad Ian R. Petersen$^{2}$, \quad Guodong Shi$^{2,3}$%
\thanks{${}^*$This work is supported by the Australian Research Council grant DP240101494.}
\thanks{$^{1,2}$School of Engineering, Australian National University, Canberra, ACT,
Australia:
        {\tt\small igor.g.vladimirov@gmail.com,
        i.r.petersen@gmail.com}.}%
\thanks{$^{3}$Australian Centre for Robotics, University of Sydney,
Camperdown, Sydney, NSW,
Australia:
        {\tt\small guodong.shi@sydney.edu.au}.}%
}

\maketitle
\thispagestyle{empty}
\pagestyle{plain}

\begin{abstract}
This paper considers a network of open quantum harmonic oscillators which interact with their neighbours through direct energy and field-mediated couplings  and also with external quantum fields. The
position-momentum dynamic variables of the network are governed by  linear quantum stochastic differential equations associated with the nodes of a graph whose edges specify the interconnection of the component oscillators. Such systems can be employed as Heisenberg picture quantum memories  with an engineered ability to approximately retain initial conditions over a
bounded  time interval. We use the
quantum  memory decoherence time defined  previously in terms of a fidelity threshold on a weighted mean-square deviation for a subset (or linear combinations)    of network variables from their initial values. This approach is applied to maximizing a high-fidelity asymptotic approximation of the  decoherence time over the direct energy coupling parameters of the network. The resulting optimality condition is a set of linear equations for blocks of a sparse matrix associated with the edges of the direct energy coupling graph of the network. We also discuss a setting where the quantum network has
a subset of dynamic variables which are affected by the external fields only indirectly,  through a complementary ``shielding'' system. This holds under a rank condition on the network-field coupling matrix and can be achieved through an appropriate field-mediated coupling between the component oscillators.
The partially isolated subnetwork has
a longer decoherence time in the high-fidelity limit, thus providing
a particularly relevant candidate for a quantum memory.
\end{abstract}

\section{INTRODUCTION}

Due to the operator-valued nature of quantum mechanics, which  describes  the dynamics and statistical properties of the physical world on atomic and subatomic scales \cite{LL_1981,S_1994},  the quantum variables are of the power of the continuum even in the finite-level case. For example, the qubit, which is the simplest quantum mechanical system,   involves the Hilbert space $\mC^2$ as opposed to the two-point state space of classical binary systems (such   as flip-flop circuits in digital computer electronics) or Bernoulli random variables. In line with the qualitatively different noncommutative structure of quantum variables and quantum states, their dynamics and quantum probabilistic properties \cite{H_2001} follow the nonclassical Heisenberg or Schr\"{o}dinger pictures, including the postulate of  unitary evolution  for  isolated systems completely specified by the Hamiltonian.

The fragility of quantum dynamics with respect to external fields and measurements  complicates the manipulation of quantum systems
needed for the initialization and preservation of particular quantum states or dynamic variables over the course of time. The latter tasks correspond to the stages of ``writing'' and storing  the quantum information (followed by the ``reading'' phase) and can be implemented   using photonics and light-matter interaction (\cite{FCHJ_2016,HEHBANS_2016,WM_2008,WLZ_2019,YJ_2014} and references therein) or other platforms. These operations involve a controlled isolation of the quantum system from its environment, similarly to   the existing paradigms of quantum computation   \cite{NC_2000}, where   the closed unitary evolution alternates with measurements.

The reversibility of unitary dynamics  (especially in the zero Hamiltonian case where any quantum state and all the system variables are preserved in time) is particularly useful for quantum information storage. However, in contrast to this idealised scenario,
the ability to retain quantum states or dynamic variables is corrupted by the unavoidable interaction of the quantum system with its environment. The system-field interaction   gives rise to quantum noise \cite{GZ_2004} which forces the system to drift away from its initial condition in a dissipative (and hence, irreversible) fashion.
The resulting open quantum dynamics are modelled  by the Hudson-Parthasarathy quantum stochastic differential equations (QSDEs) \cite{HP_1984,P_1992} driven by quantum Wiener processes. They are linear for open quantum harmonic oscillators (OQHOs) \cite{JNP_2008,NY_2017,P_2017,ZD_2022} and quasi-linear for finite-level open quantum systems \cite{EMPUJ_2016,VP_2022_SIAM}. Although the finite-level systems are particularly suitable for modelling multi-qubit registers, continuous variables systems (including the OQHOs  with their position-momentum variables) are also used in the Gaussian channel models  of quantum information theory \cite{H_2019}.

For both classes of open quantum systems in the quantum stochastic calculus framework,   their performance as a Heisenberg picture quantum memory in the storage phase has recently been described in terms of a decoherence time \cite{VP_2024_ANZCC,VP_2024_EJC}. The latter is defined as the  time horizon at which a weighted mean-square deviation of the system variables from their initial values reaches a given  fidelity threshold (see \cite{VP_2023_SCL} and references therein for different yet related decoherence measures).
The maximization of the memory decoherence time (or its asymptotic approximation in the high-fidelity or short-horizon limit) over the energy and coupling parameters of the quantum system improves its ability to approximately retain  the initial conditions over a bounded time interval. The approximate decoherence time maximization for OQHOs has also been considered in \cite{VP_2024_arxiv} for a partially isolated subsystem which is affected by the external fields only indirectly through another subsystem of the oscillator, thus leading to a qualitatively longer memory decoherence time. However, these  optimization problems   were considered in  \cite{VP_2024_ANZCC,VP_2024_arxiv} with respect to  a single direct energy coupling matrix for a feedback interconnection of two OQHOs with direct energy and field-mediated couplings \cite{ZJ_2011a}. At the same time, the development of systematic methods for solving  such problems for  multicomponent quantum networks  with arbitrary interconnection architecture is of theoretical and practical interest.

In  the present paper, we extend the approach of \cite{VP_2024_ANZCC,VP_2024_arxiv} to a quantum memory system organised as a  network of OQHOs which interact with their neighbours through direct energy and field-mediated couplings  and also with external quantum fields. The
position-momentum dynamic variables of the network are governed by  a set of cross-coupled linear QSDEs associated with the nodes of a graph whose edges specify the interconnection of the constituent oscillators. This allows the network to be represented as an augmented OQHO by making advantage of the linearity of the QSDEs (instead of the quantum feedback network formalism \cite{GJ_2009} in its full generality).
We use the
quantum  memory decoherence time defined  previously in terms of a fidelity threshold on a weighted mean-square deviation for a subset (or linear combinations)    of network variables from their initial values. This approach is applied to maximizing the high-fidelity asymptotic approximation of the  decoherence time over the direct energy coupling matrices of the network. The resulting necessary and sufficient condition of optimality is a set of linear algebraic equations (resembling the Sylvester equations \cite{GLAM_1992,SIG_1998}) for blocks of a sparse matrix associated with the edges of the direct energy coupling graph of the network. We also discuss a setting where the quantum network has
a subset of dynamic variables which are affected by the external fields only indirectly,  through a complementary ``shielding'' system. This  partial subnetwork isolation from the external fields and the related network decomposition  hold under a certain rank condition on the network-field coupling matrix and can be achieved through an appropriate field-mediated coupling between the component oscillators.
The partially isolated subnetwork acquires
a longer decoherence time in the high-fidelity limit, thus providing
a particularly relevant candidate for a quantum memory.

The paper is organised as follows.
Section~\ref{sec:net} describes the internal dynamic and external field variables of the network of OQHOs being considered.
Section~\ref{sec:par} specifies the direct energy  and field-mediated couplings between the component OQHOs (including the neighbourhoods and Hamiltonians) and their coupling to the external fields.
Section~\ref{sec:dyn} provides the set of QSDEs for the Heisenberg dynamics of the quantum network and computes the  energy and system-field coupling  matrices for the augmented OQHO which represents it.
Section~\ref{sec:mem} quantifies the performance of the network as a temporary memory system in terms of weighted deviations of its dynamic variables from their initial conditions and the memory  decoherence  time.
Section~\ref{sec:tauhatmax} formulates and solves the problem of maximizing the approximate decoherence  time in the high-fidelity limit over the direct energy coupling parameters of the network.
Section~\ref{sec:dev} specifies the above results (including the approximate decoherence time maximization)  in application to the partial subnetwork isolation setting.
Section~\ref{sec:conc} provides concluding remarks.

\section{OPEN QUANTUM NETWORK AND EXTERNAL FIELDS}
\label{sec:net}

We consider a network of linear quantum stochastic systems  at the vertex set $V$ of  a finite graph specified below. For any $j\in V$, the $j$th component system is a multi-mode open quantum harmonic oscillator (OQHO) with an even number $n_j$ of internal dynamic variables which are time-varying self-adjoint operators of position-momentum type on (a dense domain of) a Hilbert space $\fH$. These $n_j$ quantum variables are assembled into a column-vector $X_j(t)$ (the time argument $t\> 0$ will often be omitted for brevity) and act initially (at $t=0$) on a Hilbert  space $\fH_j$.   It is assumed that they satisfy the canonical commutation relations (CCRs)
\begin{equation}
\label{Xcomm}
    [X_j(t), X_j(t)^\rT] = 2i\Theta_j,
    \qquad
    [X_j(t), X_k(t)^\rT] = 0
\end{equation}
for all $t\> 0$ and  $j,k\in V$ such that $j\ne k$,
where the transpose  $(\cdot)^\rT$
applies to matrices and vectors of operators as if the latter were scalars,
$i:=\sqrt{-1}$ is the imaginary unit,
and
\begin{equation}
\label{Thetaj}
  \Theta_j := \frac{1}{2}I_{n_j/2}\ox \bJ
\end{equation}
is a nonsingular real antisymmetric matrix of order $n_j$. Here,
$[\alpha, \beta^\rT] := ([\alpha_a, \beta_b])_{1\< a\< r, 1\< b\< s} = \alpha\beta^\rT - (\beta\alpha^\rT)^\rT $ is the matrix of commutators $[\alpha_a, \beta_b] = \alpha_a \beta_b-\beta_b\alpha_a$ for vectors $\alpha:= (\alpha_a)_{1\< a\< r}$,  $\beta:= (\beta_b)_{1\< b\< s}$ of linear operators. Also, $I_r$ is the identity matrix of order $r$,  and $\ox$ is the tensor product of spaces or operators,  including the Kronecker  product of matrices as in (\ref{Thetaj}), where
\begin{equation}
\label{bJ}
    \bJ :=
\begin{bmatrix}
  0 & 1\\
  -1 & 0
\end{bmatrix}
\end{equation}
spans the subspace of antisymmetric matrices of order $2$. In view of (\ref{Xcomm}), (\ref{Thetaj}), an augmented $n$-dimensional vector $X$ of  network variables,  given by
\begin{equation}
\label{X}
  X:= (X_j)_{j\in V},
  \qquad
  n:= \sum_{j \in V}n_j,
\end{equation}
has a block-diagonal CCR matrix:
\begin{equation}
\label{XXcomm}
    [X, X^\rT] = 2i\Theta,
    \qquad
    \Theta
    :=
    \diag_{j \in V}(\Theta_j)
    =
    \frac{1}{2}
    I_{n/2}\ox \bJ.
\end{equation}
The relation (\ref{XXcomm}) holds, for example, if the network variables are implemented as the pairs $q_1, p_1, \ldots, q_{n/2},p_{n/2}$ of quantum mechanical positions $q_r$ and momenta $p_r:= -i\d_{q_r}$, with $r = 1, \ldots, n/2$,  on the  Schwartz space \cite{V_2002} of rapidly decreasing functions on $\mR^{n/2}$.

In addition to the internal variables, the $j$th OQHO has multichannel  input and output bosonic fields $W_j$, $Y_{jk}$ (see Fig.~\ref{fig:net} for an informal illustration),
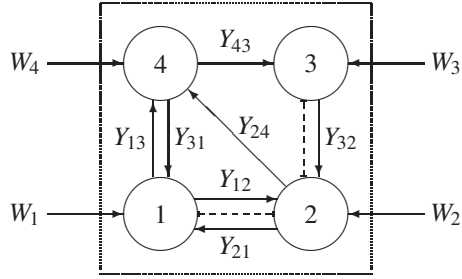
\begin{figure}[htbp]
{\centering
\unitlength=1mm
\linethickness{0.4pt}
\begin{picture}(60.00,35)
\put(12,-3){\dashbox{0.2}(36,36)}
\put(20,5){\circle{10}}
\put(40,5){\circle{10}}
\put(40,25){\circle{10}}
\put(20,25){\circle{10}}

\put(20,5){\makebox(0,0)[cc]{1}}
\put(40,5){\makebox(0,0)[cc]{2}}
\put(40,25){\makebox(0,0)[cc]{3}}
\put(20,25){\makebox(0,0)[cc]{4}}

\put(2,5){\makebox(0,0)[cc]{$W_1$}}
\put(58,5){\makebox(0,0)[cc]{$W_2$}}

\put(58,25){\makebox(0,0)[cc]{$W_3$}}
\put(2,25){\makebox(0,0)[cc]{$W_4$}}
\put(30,28){\makebox(0,0)[cc]{$Y_{43}$}}
\put(16,15){\makebox(0,0)[cc]{$Y_{13}$}}
\put(24,15){\makebox(0,0)[cc]{$Y_{31}$}}

\put(44,15){\makebox(0,0)[cc]{$Y_{32}$}}
\put(32.5,16.5){\makebox(0,0)[cc]{$Y_{24}$}}
\put(30,9.5){\makebox(0,0)[cc]{$Y_{12}$}}
\put(30,0.5){\makebox(0,0)[cc]{$Y_{21}$}}

\put(5,25){\vector(1,0){10}}
\put(25,25){\vector(1,0){10}}
\put(5,5){\vector(1,0){10}}
\put(55,5){\vector(-1,0){10}}

\put(55,25){\vector(-1,0){10}}

\put(19,10){\vector(0,1){10}}
\put(21,20){\vector(0,-1){10}}
\put(41,20){\vector(0,-1){10}}
\put(39,10){\dashbox{0.75}(0,10)}

\put(24.5,7){\vector(1,0){11}}
\put(25,5){\dashbox{0.75}(10,0)}
\put(35.5,3){\vector(-1,0){11}}
\put(36.6,8.7){\vector(-1,1){12.9}}
\end{picture}\vskip4mm}
\caption{An illustration of a network of OQHOs (depicted as numbered circles) with external quantum Wiener processes $W_j$ and output fields $Y_{jk}$ which mediate the coupling between the component OQHOs. This  field-mediated coupling is represented by arrows, while dashed lines indicate the direct energy coupling. }
\label{fig:net}
\end{figure}
which consist of $m_j$  and $r_{jk}$ time-varying  self-adjoint quantum variables, respectively (the dimensions  $m_j$, $r_{jk}$ are also even).
The input field $W_j$ is a quantum Wiener process on a symmetric Fock space $\fF_j$.  The network-field space   has the tensor-product structure
\begin{equation}
\label{fH}
    \fH
    :=
    \fH^0 \ox \fF,
    \qquad
    \fH^0
    :=
    \ox_{j\in V} \fH_j,
    \qquad
    \fF
    :=
    \ox_{j\in V} \fF_j,
\end{equation}
where the composite Fock space $\fF$ accommodates the external fields,  and the initial network space  $\fH^0$ provides a domain for the initial network variables.
Accordingly, the internal network variables and the output fields act on the network-field space $\fH$. Now, the input fields satisfy the two-point CCRs
\begin{equation}
\label{Wcomm}
  [W_j(s), W_j(t)^\rT]
  =
  2i
  \min(s,t)
  J_j,
  \qquad
  [W_j(s), W_k(t)^\rT]
  =0
\end{equation}
for all  $s,t\>0$ and $j, k\in V$ such that  $j\ne k$,
where
\begin{equation}
\label{Jj}
    J_j := I_{m_j/2}\ox \bJ
\end{equation}
is an orthogonal  real antisymmetric matrix of order $m$ defined in terms of (\ref{bJ}), so that $J_j^2 = -I_{m_j}$. The right-hand side of (\ref{Wcomm}) vanishes at $s=0$ or $t=0$ since the initial input field operators act as the identity operator on $\fF$, which commutes with any operator. Due to the continuous tensor-product     structure \cite{PS_1972} of the Fock space filtration,  the relation (\ref{Wcomm}) is equivalent to its fulfillment for all $s=t\>0$, whose incremental form is given by
\begin{equation}
\label{dWcomm}
    \rd [W_j, W_k^\rT]
     =
    [\rd W_j, W_k^\rT]
    +
    [W_j, \rd W_k^\rT]
    +
  [\rd W_j, \rd W_k^\rT]
  =
    [\rd W_j, \rd W_k^\rT]
    =
    \left\{
    \begin{matrix}
      2iJ_j \rd t  & {\rm if} & j=k\\
      0 & {\rm if} & j\ne k
    \end{matrix}
    \right..
\end{equation}
Here, the quantum Ito lemma \cite{HP_1984} is used along with the property of
the future-pointing Ito increments of the input quantum Wiener processes to commute with adapted processes (in the sense of the network-field space filtration).
In particular,
\begin{equation}
\label{dWWXYcomm}
    \big[
       \rd W_k(t),
       \begin{bmatrix}
          W_j(s)^\rT &
          X_j(s)^\rT &
          Y_j(s)^\rT
       \end{bmatrix}
    \big] = 0
\end{equation}
for all     $t\> s\> 0$ and $j,k\in V$.
We assume that the input fields are in the tensor-product vacuum state \cite{P_1992}  (in particular, they are statistically independent  for different  component OQHOs)
\begin{equation}
\label{ups}
  \ups:= \ox_{j \in V} \ups_j
\end{equation}
on the composite Fock space $\fF$, where $\ups_j$ is the vacuum state on the Fock space $\fF_j$. Accordingly, the CCRs (\ref{dWcomm}) correspond to the imaginary parts of the quantum Ito relations
\begin{equation}
\label{dWjdWk}
  \rd W_j \rd W_k^\rT =
    \left\{
    \begin{matrix}
      \Omega_j \rd t  & {\rm if} & j=k\\
      0 & {\rm if} & j\ne k
    \end{matrix}
    \right.,
    \qquad
    \Omega_j:= I_{m_j} + i J_j.
\end{equation}
In what follows, it will be convenient to use (similarly to (\ref{X}))  an augmented $m$-dimensional vector $W$ of external fields given by
\begin{equation}
\label{W}
  W:= (W_j)_{j\in V},
  \qquad
  m := \sum_{j \in V} m_j.
\end{equation}
From (\ref{Jj}), (\ref{dWjdWk}), the quantum Ito relations for $W$ take the form
\begin{equation}
\label{dWdW}
  \rd W \rd W^\rT
  = \Omega \rd t,
  \qquad
  \Omega = I_m + iJ,
\end{equation}
where
\begin{equation}
\label{J}
    J:= \diag_{j \in V}(J_j) = I_{m/2}\ox \bJ.
\end{equation}
The above described commutation and statistical structure of the external fields is essential for the dynamical and quantum probabilistic properties of the network. The latter also depend on the network-field interaction and the coupling between the component OQHOs considered in the next section.

\section{DIRECT ENERGY  AND FIELD-MEDIATED COUPLINGS}
\label{sec:par}

With every node $j \in V$ of the quantum network, we associate three subsets of the vertex set $V$ of the graph which do not contain $j$:
\begin{equation}
\label{nei}
    N_j^0, N_j^+, N_j^- \subset V\setminus \{j\},
    \qquad
    j \in V.
\end{equation}
Each of these subsets specifies a neighbourhood with which the $j$th OQHO interacts through a direct energy or field-mediated coupling.

More precisely, $N_j^0$ pertains to the direct energy coupling (for example, as between the OQHOs 1 and 2 and also between 2 and 3 in Fig.~\ref{fig:net}) and satisfies a symmetry property: for all $j,k \in V$ such that $j\ne k$,  the inclusion $k \in N_j^0$ holds if and only if so does $j \in N_k^0$. The set $N_j^0$ is therefore a neighbourhood of the node $j$ in an undirected graph (whose edges are represented by dashed lines in Fig.~\ref{fig:net}).

The set $N_j^+$  in (\ref{nei}) specifies a neighbourhood whose component OQHOs $k \in N_j^+$ ``receive'' the output fields $Y_{jk}$ from the $j$th  OQHO.
In a similar fashion, the set $N_j^-$  describes another neighbourhood whose component OQHOs $k \in N_j^-$ ``send'' their output fields $Y_{kj}$ to the $j$th  OQHO. The sets $N_j^{\pm}$ are specified by the outgoing and incoming edges of the node $j$ in a directed graph, respectively.
Accordingly, for all $j, k\in V$ such that $j\ne k$, the inclusions $k \in N_j^+$ and $j \in N_k^-$ are equivalent; see Fig.~\ref{fig:N+-}.
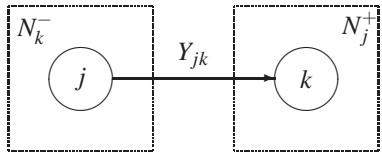
\begin{figure}[htbp]
{\centering
\unitlength=1.2mm
\linethickness{0.4pt}
\begin{picture}(60.00,13)
\put(20,5){\circle{7}}
\put(45,5){\circle{7}}
\put(12,-3){\dashbox{0.2}(16,16)}

\put(37,-3){\dashbox{0.2}(16,16)}
\put(20,5){\makebox(0,0)[cc]{$j$}}
\put(45,5){\makebox(0,0)[cc]{$k$}}

\put(15,10.5){\makebox(0,0)[cc]{$N_k^-$}}
\put(51,10.5){\makebox(0,0)[cc]{$N_j^+$}}
\put(32.5,7.5){\makebox(0,0)[cc]{$Y_{jk}$}}
\put(23.5,5){\vector(1,0){18}}
\end{picture}\vskip3mm}
\caption{Two component OQHOs $j$ and $k$ (with $j,k \in V$ such that $j\ne k$) belong to the neighbourhoods  $N_k^-$ and $N_j^+$ (shown as dashed boxes), respectively, if and only if
they are coupled through the output  field $Y_{jk}$ from $j$ to $k$. }
\label{fig:N+-}
\end{figure}
The output fields $Y_{jk}$ form  an $r_j$-dimensional  ``total'' output field $Y_j$ of the $j$th OQHO as
\begin{equation}
\label{Yj}
    Y_j
    :=
    (Y_{jk})_{k \in N_j^+},
    \qquad
            r_j
        :=
        \sum_{k \in N_j^+} r_{jk}
        \<
        m_j
\end{equation}
(the inequality is clarified below).

Ignoring the field-mediated coupling (to be considered later), the direct energy coupling is described in terms of the neighbourhoods $N_j^0$ and the following  Hamiltonian of the OQHO network:
\begin{equation}
\label{HN}
    H^0
    :=
    \frac{1}{2}
    X^\rT
    R^0
    X
    =
    \sum_{j \in V}
    \big(
    H_j + \sum_{k \in N_j^0}
    H_{jk}
    \big),
\end{equation}
where $R^0:= (R_{jk}^0)_{j,k\in V}\in \mS_n$ is the energy matrix of the network (where $\mS_n$ denotes the subspace of real symmetric matrices of order $n$), with $R_{jk}^0 :=0$ if $k \in V\setminus (\{j\} \bigcup N_j^0)$ , and use is also made of (\ref{X}). In (\ref{HN}),
\begin{equation}
\label{Hj}
    H_j := \frac{1}{2} X_j^\rT R_j X_j
\end{equation}
is the individual Hamiltonian
 of the $j$th OQHO  specified by the corresponding diagonal block  $R_j:= R_{jj}^0 \in \mS_{n_j}$ of $R^0$, while the matrix $R_{jk}^0 = R_{kj}^{0\rT} \in \mR^{n_j \x n_k}$ parameterises the direct energy coupling Hamiltonian for the $j$th and $k$th OQHOs:
 \begin{equation}
 \label{Hjk}
    H_{jk} :=
    \frac{1}{2}
    X_j^\rT R_{jk}^0X_k
    =
    H_{kj},
    \qquad
    j \in V,
    \quad
    k \in N_j^0,
\end{equation}
where the equality uses the commutativity between $X_j$ and $X_k$ from (\ref{Xcomm}) (for the same reason, $H_{jk}$ are also self-adjoint operators). Therefore, the direct energy coupling between the OQHOs $j \in V$ and $k \in N_j^0$ (indicated by an edge of the undirected graph mentioned above),  contributes $\frac{1}{2} (H_{jk} + H_{kj}) = H_{jk}$ to the total energy of the system.

In addition to the direct energy coupling between the component OQHOs through the symmetric interaction Hamiltonian terms (\ref{Hjk}),  the coupling of the $j$th OQHO to the external field $W_j$ is quantified by a matrix
$M_j\in \mR^{m_j \x n_j}$. Similarly, the coupling of the $j$th OQHO  to the output field $Y_{kj}$  (such as quantized electromagnetic radiation in the case of quantum optical \cite{WM_2008}  implementation of the network)  from the $k$th OQHO, with $k \in N_j^-$,  is specified by a matrix    $N_{jk}\in \mR^{r_{kj} \x n_j}$. In comparison with the direct energy coupling, the effect of the field-mediated coupling on the energetics of the network is less straightforward and will be recovered from the network dynamics discussed in the next section.

\section{HEISENBERG NETWORK DYNAMICS}
\label{sec:dyn}

In accordance with the direct energy and field-mediated coupling, the Heisenberg evolution of the network is governed by a set of linear quantum stochastic differential equations (QSDEs)
\begin{align}
\label{dXj}
    \rd X_j
    & =
    \big(A_j X_j + \overbrace{\sum_{k \in N_j^0} A_{jk} X_k}^{\rm direct\ energy\  coupling}\big) \rd t+
    B_j \rd W_j
    +
    \overbrace{
    \sum_{k \in N_j^-}
    E_{jk} \rd Y_{kj}}^{\rm field-mediated\ coupling},\\
\label{dYj}
    \rd Y_j
    & =
    C_j X_j \rd t  +  D_j \rd W_j,
    \qquad
    j \in V.
\end{align}
Here, the matrices
\begin{align*}
    A_j\in \mR^{n_j\x n_j},
    \quad
    A_{jk}\in \mR^{n_j\x n_k},
    \quad
    B_j\in \mR^{n_j\x m_j},
    \quad
    C_j\in \mR^{r_j\x n_j},
    \quad
    D_j\in \mR^{r_j\x m_j},
    \quad
    E_{jk}\in \mR^{n_j\x r_{kj}}
\end{align*}
are expressed in terms of the CCR matrices $\Theta_j$ from (\ref{Xcomm}), (\ref{Thetaj}), the individual energy matrices $R_j$ of the component OQHOs from  (\ref{Hj}), the direct energy coupling matrices $R_{jk}^0$ from (\ref{Hjk}), the matrices $M_j$ of coupling of the OQHOs to their external input fields and the field-mediated  coupling matrices $N_{jk}$ (see the last paragraph of the previous section) as
\begin{align}
\label{Aj}
    A_j
     & =
    2\Theta_j
    \big(
        R_j + M_j^{\rT}J_j M_j +
        \sum_{\ell \in N_j^-}
        N_{j\ell }^{\rT}
        \wt{J}_{\ell j}
        N_{j\ell }
    \big),\\
\label{AjkBj}
    A_{jk}
    & =
    2\Theta_j R_{jk}^0,
    \qquad\quad
        B_j
     = 2\Theta_j M_j^{\rT},\\
\label{CjEjk}
    C_j  & =2D_jJ_j M_j,
    \qquad
    E_{jk}  = 2\Theta_j N_{jk}^{\rT},
\end{align}
where
\begin{equation}
\label{Jkj}
    \wt{J}_{\ell j}
    :=
    D_{\ell j}J_\ell  D_{\ell j}^\rT.
\end{equation}
This parameterization (including the quadratic functions of the coupling matrices in (\ref{Aj})) reflects the physical realizability (PR) conditions  \cite{JNP_2008,SP_2012} for the coefficients of linear  QSDEs.
Furthermore, in accordance with (\ref{Yj}), (\ref{dYj}), the matrices
\begin{equation}
\label{CjDj}
    C_j
    := (C_{jk})_{k \in N_j^+},
    \qquad
    D_j := (D_{jk})_{k \in N_j^+}
\end{equation}
consist of blocks $C_{jk}\in \mR^{r_{jk}\x n_j}$ and $D_{jk}\in \mR^{r_{jk}\x m_j}$ with even numbers of rows $r_{jk}$, which,  similarly to the first equality in (\ref{CjEjk}), are related by
\begin{equation}
\label{Cjk}
  C_{jk}
  = 2D_{jk} J_j M_j.
\end{equation}
The blocks $D_{jk}$ of the matrix $D_j$ describe the selection of the output field channels from the $j$th OQHO and  their ``distribution'' over the receiving OQHOs $k \in  N_j^+$.
Accordingly, the matrix $D_j$ is formed from $r_j$
conjugate rows of a permutation matrix of order $m_j$ (which clarifies the inequality in (\ref{Yj})), so that, without loss of generality,  the quantum Ito matrix  of the output $Y_j$ of the $j$th OQHO in the relations
\begin{equation}
\label{dYjdYk}
  \rd Y_j \rd Y_k^\rT =
    \left\{
    \begin{matrix}
      D_j \Omega_jD_j^\rT \rd t  & {\rm if} & j=k\\
      0 & {\rm if} & j\ne k
    \end{matrix}
    \right.,
    \qquad
    j, k \in V
\end{equation}
has a form similar to that in (\ref{dWjdWk}):
\begin{equation}
\label{DOD}
    D_j \Omega_j D_j^\rT = I_{r_j} + iI_{r_j/2}\ox \bJ.
\end{equation}
In particular, it follows from (\ref{dWjdWk}), (\ref{Jkj}),  (\ref{CjDj}), (\ref{DOD}) that for all $j, k \in V$,
\begin{equation}
\label{DJD}
    D_{jk} J_j D_{j\ell}^\rT
    =
    \left\{
    \begin{matrix}
        \wt{J}_{jk} (= I_{r_{jk}/2}\ox \bJ)
        & {\rm if}\  k = \ell \in N_j^+ \\
      0 & {\rm otherwise}
    \end{matrix}
    \right..
\end{equation}
The block-diagonal structure of the matrix $D_j J_j D_j^\rT$ in (\ref{DOD}), (\ref{DJD}) means the commutativity between the Ito increments (at the same moment of time) of the output fields from the $j$th OQHO towards different component OQHOs in the network:
\begin{equation}
\label{dYjkdYklcomm}
    [\rd Y_{jk}, \rd Y_{j\ell}^\rT] = 0,
    \qquad
    k,\ell \in N_j^+,
    \quad
    k\ne \ell.
\end{equation}
In accordance with the partitioning in (\ref{Yj}), (\ref{CjDj}), the output $Y_{jk}$ of the $j$th OQHO towards the $k$th OQHO   satisfies the QSDE
\begin{equation}
\label{dYjk}
  \rd Y_{jk} = C_{jk} X_j \rd t + D_{jk} \rd W_j,
  \qquad
  j \in V,
  \quad
  k \in N_j^+.
\end{equation}
As the $k$th diagonal block of (\ref{dYjdYk}), the quantum Ito relations
$$
    \rd Y_{jk}\rd Y_{jk}^\rT = D_{jk} \Omega_j D_{jk}^\rT \rd t
$$
clarify the meaning of the matrix $\wt{J}_{jk}$ (cf. (\ref{Aj}), (\ref{Jkj}), (\ref{DJD})) as the CCR matrix for the output field $Y_{jk}$ in the sense that $[\rd Y_{jk}, \rd Y_{jk}^\rT] = 2i \wt{J}_{jk}\rd t$, which complements the commutativity (\ref{dYjkdYklcomm}).

The following theorem represents the dynamics (\ref{dXj}) of the quantum network  as that of a composite OQHO with the augmented vectors $X$, $W$  of network variables and input fields from (\ref{X}), (\ref{W}).

\begin{thm}
\label{thm:dyn}
The augmented vector $X$ of the network variables (\ref{X}) satisfies a linear QSDE
\begin{equation}
\label{dX}
  \rd X = A X \rd t + B \rd W
\end{equation}
driven by the augmented external field vector $W$ in (\ref{W}), where the matrices $A\in \mR^{n\x n}$ and $B \in  \mR^{n\x m}$ are given by
\begin{equation}
\label{AB}
    A := 2\Theta (R + M^\rT J M),
    \qquad
    B := 2\Theta M^\rT.
\end{equation}
Here, $\Theta$, $J$ are the augmented CCR matrices of the network and external field variables from (\ref{XXcomm}), (\ref{dWdW}), (\ref{J}). Furthermore, $M:= (M_{jk})_{j,k\in V} \in \mR^{m\x n}$ is the  network-field coupling matrix whose blocks $M_{jk}\in \mR^{m_j\x n_k}$ are computed as
\begin{equation}
\label{Mjk}
    M_{jk}
    =
    \left\{
    \begin{matrix}
      M_j & {\rm if} & j = k\\
      D_{jk}^\rT N_{kj} & {\rm if} & k \in N_j^+\\
      0 & {\rm if} & k \not\in \{j\} \bigcup N_j^+
    \end{matrix}
    \right..
\end{equation}
Also, $R:= (R_{jk})_{j,k\in V}\in \mS_n$ is the network energy matrix
with blocks $R_{jk}\in \mR^{n_j \x n_k}$ given by
\begin{align}
\label{Rjj}
    R_{jj}
    & =   R_j,\\
\label{Rjk}
  R_{jk}
  & =
  \chi_{N_j^0}(k)
  R_{jk}^0
  +
    \chi_{N_j^-}(k)M_{kj}^\rT J_k M_k
    -
    \chi_{N_j^+}(k)
    M_j^\rT J_j M_{jk}
\end{align}
for all $j, k \in V$ such that $j\ne k$, where $\chi_S$ is the indicator function of a set $S$ (that is, $\chi_S(z) = 1$ if $z \in S$ and $\chi_S(z)=0$ if $z\not \in S$).\hfill$\square$
\end{thm}
\begin{proof}
Substitution of (\ref{Aj})--(\ref{Cjk}), (\ref{dYjk}) into (\ref{dXj}) leads to
\begin{align}
\nonumber
    \rd X_j
    & =
    2\Theta_j
    \Big(
    \Big(
    \big(
        R_j + M_j^{\rT}J_j M_j +
        \sum_{\ell  \in N_j^-}
        N_{j\ell }^{\rT}
        \wt{J}_{\ell j}
        N_{j\ell }
    \big)X_j
    +
    \sum_{k \in N_j^0}
    R_{jk}^0X_k
    \Big) \rd t
    +
    M_j^{\rT}\rd W_j
    +
    \sum_{k \in N_j^-}
    N_{jk}^{\rT}D_{kj} (2J_k M_kX_k \rd t + \rd W_k)
    \Big),\\
\nonumber
    & =
    2\Theta_j
    \Big(
    \Big(
        \sum_{k \in V}
        R_{jk}^0 X_k
        +
    \big(
        \overbrace{
        M_j^{\rT}J_j M_j
        +
        \sum_{\ell  \in N_j^-}
        \underbrace{N_{j\ell }^{\rT}
        \wt{J}_{\ell j}
        N_{j\ell }}_{M_{\ell j}^\rT J_\ell M_{\ell j}}
        }^{\gamma_{jj}}
    \big)X_j
    +
        \sum_{k \in N_j^-}
            \overbrace{
        2
    \underbrace{N_{jk}^{\rT}D_{kj}}_{M_{kj}^\rT} J_k M_k}^{\gamma_{jk}}
    X_k
    \Big) \rd t
    +
    \sum_{k \in V}
    M_{kj}^\rT \rd W_k
    \Big)\\
\label{dXj1}
    & =
    2\Theta_j
    \sum_{k \in V}
    (
    (R_{jk}^0 + \gamma_{jk}) X_k \rd t
    +
    M_{kj}^\rT \rd W_k
    ),
    \qquad
    j \in V,
\end{align}
where use is also made of (\ref{Jkj}), (\ref{Mjk}) and the blocks $R_{jk}^0$ of the energy matrix $R^0$ from (\ref{HN})--(\ref{Hjk}) which does not take into account the field-mediated coupling between the component OQHOs. Here, $\Gamma:= (\gamma_{jk})_{j,k\in V} \in \mR^{n\x n}$ is an auxiliary matrix whose blocks $\gamma_{jk} \in \mR^{n_j \x n_k}$ are defined on the  right-hand side of (\ref{dXj1}):
\begin{equation}
\label{gamjk}
  \gamma_{jk}
  :=
  \left\{
  \begin{matrix}
        \sum_{\ell \in \{j\}\bigcup N_j^-}
        M_{\ell j}^\rT J_\ell M_{\ell j} & {\rm if} &  j = k\\
        2M_{kj}^\rT J_k M_k & {\rm if} & k \in N_j^-\\
        0 & {\rm if} &
        k \not\in \{j\}\bigcup N_j^-
  \end{matrix}
  \right.
\end{equation}
for all $j,k \in V$. This allows the QSDEs (\ref{dXj1}) to be assembled into one QSDE
\begin{equation}
\label{dX1}
    \rd X
    =
    2\Theta (R^0 + \Gamma)X\rd t
    +
    \underbrace{2\Theta M^\rT}_{B} \rd W,
\end{equation}
where the second notation from  (\ref{AB}) is also used. Now note that
\begin{equation}
\label{bAGam}
  \bA(\Gamma) = M^\rT J M  ,
\end{equation}
where $\bA$ is the antisymmetrizer of square matrices: $\bA(\Gamma) := \frac{1}{2}(\Gamma - \Gamma^\rT)$. Indeed,
\begin{equation}
\label{bAGamjj}
    (\bA(\Gamma))_{jj} = \bA(\gamma_{jj}) = \gamma_{jj} = (M^\rT J M)_{jj} ,
    \qquad
    j \in V
\end{equation}
in view of the antisymmetry of the diagonal blocks $\gamma_{jj}$ of $\Gamma$ in (\ref{gamjk}) inherited from the matrices $J_j$ in  (\ref{Jj}). Furthermore,
\begin{equation}
\label{bAGamjk}
    (\bA(\Gamma))_{jk}
    =
    \frac{1}{2}(\gamma_{jk} - \gamma_{kj}^\rT)
    =
    \chi_{N_j^-}(k)M_{kj}^\rT J_k M_k
    +
    \chi_{N_j^+}(k)
    M_j^\rT J_j M_{jk}
\end{equation}
for any $j,k \in V$ such that $j\ne k$.
On the other hand, (\ref{Mjk}) yields
\begin{equation}
\label{MJMjk}
    (M^\rT J M)_{jk}
    =
    \sum_{\ell \in V}
    M_{\ell j}^\rT J_\ell M_{\ell k}
    =
    \chi_{N_j^+}(k)
    M_j^\rT J_j M_{jk}
    +
    \chi_{N_j^-}(k)
    M_{kj}^\rT J_k M_k
    +
    \sum_{\ell \in (N_j^-\bigcap N_k^-)\setminus \{j,k\}}
    N_{j\ell}^\rT
    \underbrace{D_{\ell j}
    J_\ell
    D_{\ell k}^\rT}_{0} N_{k\ell}
\end{equation}
for all $j,k \in V$ such that $j\ne k$, with (\ref{DJD}) making the rightmost sum in (\ref{MJMjk}) vanish.  A comparison of (\ref{bAGamjk}) with (\ref{MJMjk}) leads to
$$    (\bA(\Gamma))_{jk} = (M^\rT J M)_{jk} ,
    \qquad
    j, k \in V,\
    j\ne k,
$$
which, in combination with (\ref{bAGamjj}), establishes (\ref{bAGam}). Due to (\ref{bAGam}),
there exists a unique matrix $R = R^\rT \in \mR^{n \x n}$ satisfying
\begin{equation}
\label{RGamma}
    R^0 + \Gamma = R + M^\rT J M.
\end{equation}
Indeed, since the energy matrix $R^0$ in (\ref{HN}) is symmetric while $M^\rT J M$ is antisymmetric, then, in view of (\ref{bAGam}), such a matrix $R$ is found by taking the symmetric part on both sides of (\ref{RGamma}):
\begin{equation}
\label{RR0Gam}
    R = R^0 + \bS(\Gamma),
\end{equation}
where $\bS$ is the symmetrizer of square matrices: $\bS(\Gamma) := \frac{1}{2}(\Gamma + \Gamma^\rT)$, with the equality of the antisymmetric parts of (\ref{RGamma}) being secured by (\ref{bAGam}). In view of (\ref{RGamma}), the QSDE (\ref{dX1}) takes the form (\ref{dX}) with the matrices (\ref{AB}). It now remains to compute the blocks of the matrix $R$ in (\ref{RR0Gam}). By the antisymmetry of $\gamma_{jj}$ in (\ref{bAGamjj}), the diagonal blocks of $R$ coincide with those of $R^0$:
$$
    R_{jj} = R_j + \bS(\gamma_{jj}) = R_j,
    \qquad
    j \in V,
$$
thus establishing (\ref{Rjj}).
The off-diagonal blocks of $R$ are found by using those of $\bS(\Gamma)$:
\begin{equation}
\label{bSGamjk}
    (\bS(\Gamma))_{jk}
    =
    \frac{1}{2}(\gamma_{jk} + \gamma_{kj}^\rT)
    =
    \chi_{N_j^-}(k)M_{kj}^\rT J_k M_k
    -
    \chi_{N_j^+}(k)
    M_j^\rT J_j M_{jk}
\end{equation}
for all $j, k\in V$ such that $j\ne k$,
where (\ref{gamjk}) is used similarly to (\ref{bAGamjk}).  A combination of (\ref{RR0Gam}) with  (\ref{bSGamjk}) establishes (\ref{Rjj}), (\ref{Rjk}), thus completing the proof.
\end{proof}

The above proof (which uses the linearity of the component QSDEs instead of the general quantum feedback network formalism \cite{GJ_2009}) clarifies the role of the quadratic terms in (\ref{Aj}) for the property (\ref{bAGam}) (including (\ref{bAGamjj})). The latter is crucial for the existence and uniqueness  of the symmetric energy matrix $R$ in the representation (\ref{AB}) leading to the PR property
$$
    A \Theta + \Theta A^\rT + BJB^\rT = 0,
$$
which corresponds to the preservation of the CCRs (\ref{XXcomm}) for the internal network variables  over the course of time \cite{JNP_2008,SP_2012} and also substantially uses the commutativity (\ref{dWWXYcomm}).

Also note that the individual energy matrices $R_j$ of the component OQHOs  enter the network energy matrix $R$ only through its diagonal blocks (\ref{Rjj}). Furthermore, in addition to the direct energy coupling between the OQHOs, the network-field coupling and the field-mediated coupling contribute to $R$ only through its off-diagonal blocks (\ref{Rjk}). The resulting network Hamiltonian $H$ is related to $H^0$ in (\ref{HN}) as
$$
    H
     = \frac{1}{2}X^\rT R X
     = H^0
    +
    \frac{1}{2}
    \sum_{j,k \in V}
    X_j^\rT
        (\chi_{N_j^-}(k)M_{kj}^\rT J_k M_k
    -
    \chi_{N_j^+}(k)
    M_j^\rT J_j M_{jk}   )
    X_k
    =
    H^0
    -
    \sum_{j \in V}
    X_j^\rT
    M_j^\rT J_j
    \sum_{k\in N_j^+}
    D_{jk}^\rT N_{kj}
    X_k,
$$
where use is also made of the blocks (\ref{Mjk}) of the network-field coupling matrix $M$ which does not depend on the  matrix $R^0$.

The QSDE (\ref{dX}), coming from the internal direct energy and field-mediated coupling in the quantum network and its interaction with the external fields, makes the network variables evolve in time $t\> 0$  as
\begin{equation}
\label{Xsol}
    X(t)
    =
    \re^{tA} X(0)
    +
    Z(t),
    \qquad
    Z(t)
    :=
    \int_0^t
    \re^{(t-s)A}
    B
    \rd W(s),
\end{equation}
with the matrices $A$, $B$ from (\ref{AB}). The fact that
$\re^{tA}$ is a usual matrix exponential makes the quantum trajectory $X$ in (\ref{Xsol}) (despite the operator-valued nature of its component processes) similar to the solution of  classical linear SDEs. This resemblance is due to the linearity of the QSDE (\ref{dX}).

\section{QUANTUM NETWORK AS TEMPORARY MEMORY SYSTEM}
\label{sec:mem}

The linear network response $\re^{tA} X(0)$ to the initial network variables in (\ref{Xsol}) can be regarded as a ``useful signal'' carrying information about $X(0)$.  It is particularly so on time scales $t \ll \frac{1}{\|A\|}$, where $\|\cdot \|$ is the operator norm of a matrix. At such time horizons, the deviation of $\re^{tA} X(0)$ from the initial condition $X(0)$ is relatively small in any suitable sense. This allows the network to be used as a temporary memory (in the Heisenberg picture) for storing (at least approximately) $X(0)$. However, the network response to the initial condition in (\ref{Xsol}) is ``corrupted'' by the zero-mean Gaussian quantum Ito process \cite[Section~3]{VPJ_2018a} $Z:= (Z_j)_{j \in V}$ adapted to the filtration of the Fock space $\fF$ in (\ref{fH}) and driven by the quantum Wiener process $W$ in the vacuum state (\ref{ups}), with $Z(0)=0$.  Note that the process $Z$ commutes with the components of $X(0)$ (since the latter act on the different initial network space $\fH_0$ in (\ref{fH})):
\begin{equation}
\label{XZcomm}
    [X(0),Z(t)^\rT] = 0,
    \qquad
    t \> 0 .
\end{equation}
Furthermore, if the network-field quantum state on the space $\fH$ in (\ref{fH}) also has a tensor-product structure given by
\begin{equation*}
\label{rho}
    \rho:= \rho_0\ox \ups,
\end{equation*}
with $\rho_0$ the initial network state on $\fH_0$, then, in addition to (\ref{XZcomm}), the process  $Z$ is statistically independent  of $X(0)$. In particular,
\begin{equation}
\label{EXZ}
  \bE(X(0)Z(t)^\rT)
  =
  \bE X(0)
  \underbrace{\bE Z(t)^\rT}_{0}
  =
  0,
\end{equation}
where $\bE\zeta := \Tr (\rho \zeta) $ is the expectation of an operator $\zeta$ on the network-field space $\fH$.

The above commutation and statistical properties can be used in order to quantify the ability of the quantum network  as a memory system to retain the initial condition of a  quantum process
\begin{equation}
\label{phiX}
    \varphi(t):= F X(t) = \sum_{j \in V} F_j X_j(t),
    \qquad
    t \> 0.
\end{equation}
Here, $F \in \mR^{s\x n}$ is a given constant matrix which has  full row rank
\begin{equation}
\label{Frank}
    \rank F = s \< n
\end{equation}
and is partitioned into blocks $F_j \in \mR^{s \x n_j}$ in accordance with (\ref{X}):
\begin{equation}
\label{FFF}
  F := (F_j^\rT)_{j \in V}^\rT.
\end{equation}
The rows of $F$ specify the coefficients  of $s$ independent linear combinations (in particular, a subset) of the network variables which need (or have a better chance) to be preserved over a bounded time interval. Note that, the process $\varphi$ in  (\ref{phiX}), in general,   involves the internal variables of different OQHOs in the network and thus corresponds to a spatially distributed  storage of quantum information.
The deviation of $\varphi$ from its  initial value is related by
\begin{equation}
\label{eta}
    \eta(t):= \varphi(t) - \varphi(0) = F \xi(t)
\end{equation}
to a similar deviation for the underlying network process $X$:
\begin{equation}
\label{xi}
     \xi(t)
    :=
    X(t)-X(0)
    =
    \alpha_t X(0) + Z(t).
\end{equation}
Here, the last equality follows from (\ref{Xsol}) and employs an auxiliary matrix
\begin{equation}
\label{exp}
  \alpha_t := \re^{tA}-I_n.
\end{equation}
Following \cite{VP_2024_ANZCC,VP_2024_EJC},  we will describe the quantum network memory performance in terms of a quadratic form
\begin{equation}
\label{Q}
    Q(t)
    :=
    \eta(t)^\rT \eta(t)
    =
    \xi(t)^\rT \Sigma\xi(t),
\end{equation}
where, in accordance with (\ref{eta}), the weighting matrix $0\preccurlyeq\Sigma \in \mS_n$  is factorised as
\begin{equation}
\label{FF}
    \Sigma := F^\rT F,
\end{equation}
so that $\rank \Sigma = s$ due to (\ref{Frank}).
Note that $\xi(0) = 0$, $\eta(0)=0$, $\alpha_0 = 0$ and  $Q(0) = 0$ in view of (\ref{eta})--(\ref{Q}). At an arbitrary moment of time $t\> 0$,   the   ``size'' of the deviation $\eta(t)$ in  (\ref{eta})   can be quantified by a mean-square functional
\begin{equation}
\label{Del}
    \Delta(t)
     :=
    \bE
    Q(t)
    =
    \bra
    \Sigma,
    \Re \Ups(t)
    \ket
    =
    \|F \alpha_t \sqrt{P}\|^2 + \bra \Sigma, \Re \Lambda(t) \ket
\end{equation}
using (\ref{Q}),  as
proposed in  \cite{VP_2024_ANZCC,VP_2024_EJC}, where $\|\cdot\|$, $\bra\cdot, \cdot\ket $   are the Frobenius norm and inner product   of matrices \cite{HJ_2007}, respectively. Here,
$$
    \Ups(t)
     :=
    \bE (\xi(t)\xi(t)^\rT)
    =
    \alpha_t
    \Pi
    \alpha_t^\rT
    +
    \Lambda(t)
$$
 is the one-point second-moment matrix  of the process $\xi$
found in \cite{VP_2024_ANZCC} by using (\ref{XZcomm}), (\ref{EXZ}), (\ref{xi}), (\ref{exp}) and the second-moment matrix
\begin{equation}
\label{EXX0}
    \Pi
     :=
    \bE (X(0)X(0)^\rT)
    =
    P + i\Theta,
    \qquad
    P:= \Re \Pi
\end{equation}
of the initial network variables in $X(0)$, including  their augmented CCR matrix $\Theta$ from (\ref{XXcomm}),
and the quantum covariance matrix
\begin{equation}
\label{V}
    \Lambda(t)
      :=
    \bE(Z(t)Z(t)^\rT)
     =
    \int_0^t
    \re^{sA}
    B
    \Omega
    B^\rT
    \re^{sA^\rT}
    \rd s,
    \quad
    t \> 0
\end{equation}
of the process $Z$ in (\ref{Xsol}), with
$\Omega$ the quantum Ito matrix from (\ref{dWdW}). The matrix $\Lambda(t)$ in (\ref{V}) coincides with the controllability Gramian of the pair $(A, B\sqrt{\Omega})$ over the time interval $[0,t]$ and   satisfies  the initial value problem for the Lyapunov ODE
\begin{equation}
\label{VALE}
    \mathop{\Lambda}^{\centerdot}(t)
    =
    A\Lambda(t) + \Lambda(t)A^\rT + \mho,
    \qquad
    t \> 0,
    \quad
    \Lambda(0) = 0,
\end{equation}
where
\begin{equation}
\label{BOB}
    \mho := B\Omega B^\rT.
\end{equation}
From (\ref{VALE}), the time derivatives $\Lambda^{(k)}:= \rd^k \Lambda/\rd t^k$ can be computed recursively as
$
    \Lambda^{(k+1)} = A\Lambda^{(k)} + \Lambda^{(k)}A^\rT  + \delta_{k0} \mho$ for any $
    k \> 0
$,
where $\delta_{jk}$ is the Kronecker delta. In particular, the first three derivatives at $t=0$  are
\begin{align}
\label{V1}
    \dot{\Lambda}(0)
    & = \mho,\\
\label{V2}
    \ddot{\Lambda}(0)
    & = A\mho  + \mho A^\rT ,\\
\label{V3}
    \dddot{\Lambda}(0)
    & =
    A^2\mho + \mho (A^\rT)^2 + 2 A \mho A^\rT.
\end{align}
In \cite{VP_2024_ANZCC},  the relations (\ref{BOB})--(\ref{V3}) were used in order to compute
the following time derivatives of $\Delta(t)$ in (\ref{Del}) at $t=0$:
\begin{align}
\label{Deldot}
  \dot{\Delta}(0)
  & =
   \|F B\| ^2,\\
\label{Delddot}
   \ddot{\Delta}(0)
   & =
   \bra \Sigma, ABB^\rT+BB^\rT A^\rT + 2A P A^\rT\ket ,
\end{align}
thus leading to the short-horizon asymptotic behaviour of the function
\begin{equation}
\label{Delasy0}
    \Delta(t)
     =
     \dot{\Delta}(0)
     t +
     \frac{1}{2}
     \ddot{\Delta}(0)
     t^2
     +
     O(t^3),
     \qquad
     {\rm as}\
     t \to 0+.
\end{equation}

For improving the performance of the network as a quantum memory system for storing the initial condition $\varphi(0)$ of the process (\ref{phiX}), it is beneficial to minimize the mean-square deviation functional $\Delta(t)$ in (\ref{Del}) at a suitable time $t> 0$. Such minimization can be carried out over the energy and coupling parameters of the network. This problem  can also be approached through maximizing the  memory decoherence time \cite{VP_2024_ANZCC,VP_2024_EJC}
\begin{equation}
\label{tau}
    \tau(\eps)
    :=
    \inf
    \{
        t\> 0:\
        \Delta(t)
        \>
        \eps
        \Delta_*
    \}
\end{equation}
as a horizon by which the subnetwork variables do not deviate from their initial values
``too far''. The threshold value $\eps \Delta_*$ in (\ref{tau}) for the mean-square deviation functional (\ref{Del}) involves
\begin{equation}
\label{Ephiphi}
    \Delta_*
    :=
    \bE (\varphi(0)^\rT \varphi(0) ) = \bra \Sigma, \Pi\ket = \|F\sqrt{P}\|^2
\end{equation}
as a reference scale (obtained from (\ref{phiX}),  (\ref{FF}), (\ref{EXX0})) and a dimensionless fidelity parameter $\eps>0$. The latter specifies a relative error threshold  for $\varphi(t)$ to approximately reproduce $\varphi(0)$. While $\tau(0)=0$ by   $\Delta(0)=0$ in (\ref{Del}), we eliminate from consideration the trivial case $\Delta_* = 0$  when $\tau(\eps) = 0$ for any $\eps>0$ by assuming  that
\begin{equation}
\label{Fpos}
    F
    \sqrt{P}
    \ne 0.
\end{equation}
Under this  condition, the decoherence time $\tau(\eps)$ in (\ref{tau}), which  is a nondecreasing function of the fidelity parameter $\eps$, satisfies $\tau(\eps)>0$ for any $\eps>0$.

The condition (\ref{Fpos}) is relevant if $s<n$ when, in  view of (\ref{Frank}),   the columns of the matrix $F$ are linearly dependent. Since
\begin{equation}
\label{phiphicomm}
    [\varphi(t), \varphi(t)^\rT]
    =
    2i F\Theta F^\rT
\end{equation}
by (\ref{XXcomm}), (\ref{phiX}), the fulfillment of $F \sqrt{P} = 0$ in (\ref{Ephiphi}) implies that $i F\Theta F^\rT = F\Pi F^\rT \succcurlyeq 0 $ and hence, $F\Theta F^\rT=0$, in which case,    the components of  $\varphi(t)$ commute with each other at any time $t\> 0$. Here, use is made of the fact that a purely imaginary antisymmetric matrix is positive semi-definite if and only if it is zero. Therefore, if the vector $\varphi(0)$, to be retained by the network,  is essentially quantum in the sense that it contains at least one pair of noncommuting variables, that is,
\begin{equation}
\label{FOF}
  F\Theta F^\rT
  =
  \sum_{j \in V}
  F_j \Theta_j F_j^\rT
  \ne 0
\end{equation}
in (\ref{phiphicomm}),
then the condition (\ref{Fpos}) is necessarily satisfied, which makes (\ref{FOF}) sufficient for (\ref{Fpos}).

\section{APPROXIMATE MEMORY DECOHERENCE TIME MAXIMIZATION}
\label{sec:tauhatmax}

By using the infinite differentiablity of $\Delta(t)$ in (\ref{Del}) with respect to $t\> 0$ along with its time derivatives (\ref{Deldot}), (\ref{Delddot}), it was established in \cite{VP_2024_ANZCC}  that the decoherence time $\tau(\eps)$ in (\ref{tau})
has
the following first and second right derivatives at $\eps=0$:
\begin{align}
\label{ratio}
    \tau'(0)
    & =
    \frac{\|F\sqrt{P}\| ^2}{\dot{\Delta}(0)}
    =
    \frac{\|F\sqrt{P}\| ^2}{\|FB\| ^2}>0, \\
\label{tau2}
    \tau''(0)
     &  =
    -
    \frac{\ddot{\Delta}(0) \tau'(0)^2}
    {\dot{\Delta}(0)}
    =
    -
    \frac{\bra \Sigma, A BB^\rT + BB^\rT A^\rT  + 2APA^\rT\ket
    \|F\sqrt{P}\|^4}
    {\|FB\|^6}.
\end{align}
Here, in addition to (\ref{Fpos}) (or its stronger version (\ref{FOF})), it is assumed that
\begin{equation}
\label{pos1}
    FB \ne 0.
\end{equation}
Despite having the physical dimension of time,  both quantities $\tau'(0)$ and $\tau''(0)$  in (\ref{ratio}), (\ref{tau2}) resemble the signal-to-noise ratio since their numerators involve $F PF^\rT$ pertaining to the initial condition $\varphi(0)$ to be stored by the network (see (\ref{phiX}), (\ref{EXX0})), while $FBB^\rT F^\rT = FB\Re \Omega B^\rT F^\rT$ in the denominators is associated with the quantum diffusion matrix in
$$
    \rd \varphi \rd \varphi^\rT = FB\Omega B^\rT F^\rT \rd t
$$
coming from the network-field  coupling in view of (\ref{dWdW}), (\ref{AB}). Therefore, under the condition (\ref{pos1}), the asymptotic behaviour of $\tau(\eps)$  in the high-fidelity  limit (of small values of $\eps>0$)  can be described by an appropriately truncated Taylor series approximation \begin{equation}
\label{tau12}
    \tau(\eps) =
      \wh{\tau}(\eps)
     + O(\eps^3),
     \qquad
       \wh{\tau}(\eps)
  :=
  \tau'(0) \eps + \frac{1}{2}\tau''(0) \eps^2,
\end{equation}
as $\eps \to 0+$.
Note that the quantity $\tau'(0)$ in (\ref{ratio}) is independent of the matrix $A$, whereas $\tau''(0)$ in (\ref{tau2})  depends on both matrices $A$ and $B$,
with its dependence on $A$ being concave quadratic. As observed in \cite{VP_2024_ANZCC},   these  properties are beneficial for an effective solution of the approximate decoherence time maximization problem
\begin{equation}
\label{tauhatmax}
  \wh{\tau}(\eps)
  \longrightarrow
  \sup
\end{equation}
(for a given small value of $\eps$)
over a suitable set of those energy and coupling parameters of the network that can be varied.
More precisely, for fixed network-field and field-mediated coupling, when  the matrices $M_j$, $D_{jk}$, $N_{kj}$ are fixed for all $j \in V$ and $k \in N_j^+$, and hence,  so is the composite network-field coupling matrix $M$ in (\ref{Mjk}), then the matrix $B$ in (\ref{AB}) is also fixed. If, in addition, all of the individual energy matrices $R_j$ of the component OQHOs are also fixed, then the matrix $A$ in (\ref{AB}) is an affine function of the network energy matrix $R$ and thus of the direct energy coupling matrices $R_{jk}^0$ in  view of (\ref{Rjk}). This allows the maximization problem (\ref{tauhatmax}) over the matrices $R_{jk}^0$ to be reduced  to a convex quadratic minimization  problem in  the following theorem which is a network extension of \cite[Theorems 1, 2]{VP_2024_ANZCC}.  In order to formulate it, we define, for any $j \in V$ and $k \in N_j^0$,   a linear operator $g_{jk}$ on $\mR^{n_j \x n_k}$  by
\begin{equation}
\label{gjk}
    g_{jk}(N)
     :=
    \Theta_j \Sigma_{jj}\Theta_j N P_{kk}   +
    P_{jj}N \Theta_k \Sigma_{kk}\Theta_k
     +   \Theta_j \Sigma_{jk}\Theta_k N^\rT  P_{jk}
    +
    P_{jk}N^\rT \Theta_j \Sigma_{jk}\Theta_k
    =
    (g_{kj}(N^\rT))^\rT,
    \qquad
    N \in \mR^{n_j \x n_k},
\end{equation}
where $\Theta_j$ are the CCR matrices from (\ref{Xcomm}), (\ref{Thetaj}), and use is made of the appropriate blocks $(\cdot)_{jk}$ of the matrices $\Sigma$, $P$ from (\ref{FF}), (\ref{EXX0}). Also, $K_{jk}: \mS_n \to \mR^{n_j\x n_k}$  is an affine map given by
\begin{equation}
\label{KLjk}
    K_{jk}(R):=
    L_{jk}
    +
    2
    (\bS
    (\Theta \Sigma \Theta RP))_{jk},
\end{equation}
with the matrix
\begin{equation}
\label{Ljk}
    L_{jk}
    :=
      \frac{1}{2}
    (\bS
    (\Theta \Sigma B(B^\rT + 2J MP)))_{jk}
\end{equation}
being a constant term (which does   not depend on $R$) in (\ref{KLjk}).

\begin{thm}
\label{th:tauhatmax1}
Suppose the network-field and field-mediated coupling matrices $M_j$, $N_{jk}$ are fixed along with the matrices $D_{jk}$ and the individual energy matrices $R_j$ of the component OQHOs in the network. Also,
suppose
the fidelity level $\eps$ in (\ref{tau}) and the matrix $F$ in (\ref{phiX}), (\ref{FF})
are fixed together with the matrix
$P$  in (\ref{EXX0}), and the conditions (\ref{Fpos}), (\ref{pos1}) are satisfied. Then the direct energy coupling matrices $R_{jk}^0$  deliver a solution to the
problem
\begin{equation}
\label{taumaxRjk}
    \wh{\tau}(\eps)
    \longrightarrow
    \sup,
    \qquad
    R_{jk}^0=R_{kj}^{0\rT} \in \mR^{n_j\x n_k},
    \quad
    j \in V,
    \quad
    k \in N_j^0
\end{equation}
of maximizing the approximate decoherence time from (\ref{tau12}) if
and only if
\begin{equation}
\label{Rjksol}
    g_{jk}(R_{jk}^0) + K_{jk}(\breve{R}_{jk}) = 0,
\end{equation}
where the maps $g_{jk}$, $K_{jk}$ are given by (\ref{gjk})--(\ref{Ljk}).
Here, the matrix $\breve{R}_{jk}\in \mS_n$ is obtained by letting $R_{jk}^0 = R_{kj}^{0\rT} := 0$ in (that is, removing $R_{jk}^0$ and $R_{kj}^0$ from) the network energy matrix $R$ in (\ref{Rjj}), (\ref{Rjk}), so that $K_{jk}(\breve{R}_{jk})$ does not depend on $R_{jk}^0$.
\hfill$\square$
\end{thm}

\begin{proof}
Since $\tau'(0)$ in (\ref{ratio}) does not depend on the network energy matrix $R$, the maximization (\ref{taumaxRjk})  of $\wh{\tau}(\eps)$ from (\ref{tau12}) reduces to that of $\tau''(0)$:
\begin{equation}
\label{supsup}
    \sup_{R_{jk}^0} \wh{\tau}(\eps)
    =
    \tau'(0) \eps
    +
    \frac{1}{2}
    \eps^2
    \sup_{R_{jk}^0}  \tau''(0).
\end{equation}
In turn, (\ref{tau2}) allows the rightmost supremum in (\ref{supsup}) to be represented as
\begin{equation}
\label{supinf}
    \sup_{R_{jk}^0}  \tau''(0)
    =
    -
    \frac{    \|F\sqrt{P}\|^4}
    {\|FB\|^6}
    \inf_{R_{jk}^0}  \ddot{\Delta}(0).
\end{equation}
By (\ref{supsup}) and (\ref{supinf}), the optimization problem (\ref{taumaxRjk}) is equivalent to
\begin{equation}
\label{Delmin}
  \ddot{\Delta}(0)
  \longrightarrow \inf,
    \qquad
    R_{jk}^0=R_{kj}^{0\rT} \in \mR^{n_j\x n_k},
    \quad
    j \in V,
    \quad
    k \in N_j^0,
\end{equation}
which is a convex quadratic  minimization problem in view of (\ref{Delddot}) and the affine
dependence of $A$ in (\ref{AB}) on $R$, with $R$ being an affine function of the direct energy coupling matrices $R_{jk}^0$. Therefore, the first-order optimality condition
\begin{equation}
\label{opt}
    \d_{R_{jk}^0} \ddot{\Delta}(0) = 0,
    \qquad
    \
    j \in V,\
    k \in N_j^0,
\end{equation}
for the partial Frechet derivatives \cite{RS_1980} of $\ddot{\Delta}(0)$ with respect to the matrices $R_{jk}^0$  is both necessary and sufficient for them to deliver a global  minimum to  $\ddot{\Delta}(0)$. As obtained in the proof of \cite[Theorem~1]{VP_2024_ANZCC}, the Frechet derivative of $\ddot{\Delta}(0)$ with respect to the whole network energy matrix $R$ is given by
\begin{equation}
\label{dDel/dR}
    \d_R \ddot{\Delta}(0)
    = -4\bS(\Theta \Sigma(BB^\rT + 2AP)).
\end{equation}
The first  variation of   the composite function $R_{jk}^0 \mapsto R \mapsto  \ddot{\Delta}(0)$ can now be computed by applying the chain rule as
\begin{equation}
\label{deltaDelRjk}
    \delta \ddot{\Delta}(0)
    =
    \bra
        \d_R \ddot{\Delta}(0),
        \delta  R
    \ket
    =
    \bra
        (\d_R \ddot{\Delta}(0))_{jk},
        \delta R_{jk}^0
    \ket
    +
    \bra
        (\d_R \ddot{\Delta}(0))_{kj},
        \delta R_{jk}^{0\rT}
    \ket
    =
    2
    \bra
        (\d_R \ddot{\Delta}(0))_{jk},
        \delta R_{jk}^0
    \ket.
\end{equation}
Here, use is also made of the symmetry of the matrix $\d_R \ddot{\Delta}(0)$ in (\ref{dDel/dR}) along with the relation $\bra \alpha^\rT, \beta^\rT \ket =  \bra \alpha, \beta\ket $ for identically dimensioned matrices. A combination of
(\ref{deltaDelRjk}) with (\ref{dDel/dR}) yields the partial Frechet derivatives in (\ref{opt}):
\begin{equation}
\label{dDel/dRjk}
    \d_{R_{jk}^0}
    \ddot{\Delta}(0)
     =
    2(\d_R \ddot{\Delta}(0))_{jk}
    =
    -8(\bS(\Theta \Sigma(BB^\rT + 2AP)))_{jk}.
\end{equation}
The right-hand side of (\ref{dDel/dRjk}) is an affine function of $R$,  which can be split into $R_{jk}^0$-independent and $R_{jk}^0$-dependent parts as
\begin{equation}
\label{dDel/dRjk1}
    -\frac{1}{16}
    \d_{R_{jk}^0}
    \ddot{\Delta}(0)
    =
    \frac{1}{2}(\bS(\Theta \Sigma(BB^\rT + 2AP)))_{jk}
    =
    \frac{1}{2}(\bS(\Theta \Sigma(BB^\rT + 2\breve{A}_{jk}P)))_{jk}
    +
    g_{jk}(R_{jk}^0)
    =
    K_{jk}(\breve{R}_{jk})     + g_{jk}(R_{jk}^0).
\end{equation}
Here, the linear map $g_{jk}$ from (\ref{gjk})  is evaluated at the matrix $R_{jk}^0$, and
\begin{equation}
\label{Kjk}
    K_{jk}(\breve{R}_{jk})
    :=
    \frac{1}{2}
    (\bS
    (\Theta \Sigma (BB^\rT + 2\breve{A}_{jk}P))_{jk},
\end{equation}
with
\begin{equation}
\label{Ajkhat}
    \breve{A}_{jk}
    :=
    2\Theta
    (
    \breve{R}_{jk}
    + M^\rT JM
    )
    =
    2\Theta \breve{R}_{jk} + BJM
\end{equation}
in accordance with (\ref{AB}) and with the matrix $\breve{R}_{jk}$ being obtained by removing $R_{jk}^0$ (and also $R_{kj}^0 = R_{jk}^{0\rT}$) from $R$ in (\ref{Rjk}).  Note that the relations (\ref{Kjk}), (\ref{Ajkhat}) are in agreement with (\ref{KLjk}), (\ref{Ljk}).
By (\ref{dDel/dRjk1}), the optimality condition (\ref{opt}) takes the form (\ref{Rjksol}), thus establishing the latter as a necessary and sufficient condition of optimality for the problem (\ref{Delmin}) and its equivalent (\ref{taumaxRjk}).
\end{proof}

The linear operator $g_{jk}$,  defined by (\ref{gjk}) on the Hilbert space $\mR^{n_j \x n_k}$ (with the Frobenius inner product $\bra\cdot, \cdot \ket$),  is self-adjoint and negative semi-definite. Indeed, in view of (\ref{dDel/dRjk1}),
\begin{equation}
\label{gjkneg}
    g_{jk} =
    -\frac{1}{16}
    \d_{R_{jk}^0}^2
    \ddot{\Delta}(0)
    =
    g_{jk}^\dagger
    \preccurlyeq
    0
\end{equation}
since the second Frechet derivative $\d_{R_{jk}^0}^2
    \ddot{\Delta}(0)$ of the convex quadratic function $R\mapsto  \ddot{\Delta}(0)$ in (\ref{Delddot}) is a positive semi-definite self-adjoint operator. Therefore, if the operators $g_{jk}$ are all negative definite (and hence, have well-defined inverses $g_{jk}^{-1}$), then the optimality condition (\ref{Rjksol}) can be represented as
\begin{equation}
\label{Rjksol1}
    R_{jk}^0 = - g_{jk}^{-1}(K_{jk}(\breve{R}_{jk})),
    \qquad
     j\in V, \
     k \in N_j^0.
\end{equation}
By using the vectorization $\vec{(\cdot)}$ (or $\col(\cdot)$)  of matrices \cite{M_1988},  the action of the linear operator $g_{jk}$ in (\ref{gjk}) can be expressed as
\begin{equation}
\label{gQ}
    \col(g_{jk}(N))
     =
     Q_{jk} \vec{N}
\end{equation}
in terms of an auxiliary matrix $0\succcurlyeq Q_{jk} \in \mS_{n_jn_k}$ given by
\begin{equation}
\label{Qjk}
  Q_{jk}
  :=
      P_{kk} \ox (\Theta_j \Sigma_{jj}\Theta_j)     +
    (\Theta_k \Sigma_{kk}\Theta_k)\ox P_{jj}
    +   (P_{kj}\ox (\Theta_j \Sigma_{jk}\Theta_k)
    +
    (\Theta_k \Sigma_{kj}\Theta_j) \ox P_{jk})\bT_{jk}.
\end{equation}
Here,
$\bT_{jk}$ is a permutation matrix of order $n_j n_k$ which represents the matrix transpose in $\mR^{n_j \x n_k}$ in vectorized  form: $\col(N^\rT) = \bT_{jk} \vec{N}$ for any $N \in \mR^{n_j\x n_k}$. Hence, in view of (\ref{gjkneg}), the negative definiteness of $g_{jk}$ is equivalent to $\det Q_{jk}\ne 0$. In this case, (\ref{gQ}), (\ref{Qjk}) allow (\ref{Rjksol1}) to be vectorized as
\begin{equation}
\label{Rjksol2}
    \vec{R}_{jk}^0 = - Q_{jk}^{-1} \col(K_{jk}(\breve{R}_{jk})),
    \qquad
     j\in V, \
     k \in N_j^0.
\end{equation}
Note that (\ref{Rjksol}) describes a set of linear algebraic equations (with a resemblance to the Sylvester equations \cite{GLAM_1992,SIG_1998}) for blocks $R_{jk}^0$  of a sparse symmetric matrix whose sparsity structure is specified by the edges of the direct energy coupling graph of the quantum network.
 In the case of $g_{jk}\prec 0$,  the alternative form  (\ref{Rjksol1}) of this set of equations  (including its vectorized representation (\ref{Rjksol2})) allows it to be approached as a fixed-point problem whose practical solution can take into account the specific ``spatially distributed'' structure. Also note that  the network-field and field-mediated coupling matrices $M_j$, $N_{jk}$  enter the set of equations (\ref{Rjksol})  through the matrices (\ref{Ljk}) (which depend on $M$ from (\ref{Mjk}) in a quadratic fashion in view of the second equality in (\ref{AB})) and also through the matrices $\breve{R}_{jk}$ in view of (\ref{Rjk}).

\section{PARTIAL SUBNETWORK ISOLATION}
\label{sec:dev}

As obtained in \cite{VP_2024_arxiv}, under a certain rank condition,   the matrix $F$ in (\ref{phiX})  can be found so as to select  those linear combinations (or a subset) of the network variables  which pertain to a partially isolated subnetwork whose dynamics are affected by the external fields only indirectly. This is specified by the following lemma which is an appropriate adaptation  (with an almost verbatim proof)   of \cite[Lemma~1]{VP_2024_arxiv}.

\begin{lem}
\label{lem:iso}
Suppose the network-field coupling matrix $M$ in (\ref{Mjk}) satisfies
\begin{equation*}
\label{Mrank}
    d:= n-\rank M>0.
\end{equation*}
Then for any $s\< d$,
there exists a full row rank matrix $F \in \mR^{s\x n}$ such that
\begin{equation}
\label{FB0}
    FB =  0
\end{equation}
for the matrix $B$ from (\ref{AB}), and the
process $\varphi$ in (\ref{phiX}) satisfies
\begin{equation}
\label{phidot}
    \mathop{\varphi}^{\centerdot}(t) = G X(t).
\end{equation}
Here, the matrix $G \in \mR^{s\x n}$ is given by
\begin{equation}
\label{G}
    G := 2F\Theta R,
\end{equation}
where $\Theta$ is the CCR matrix for the network variables from (\ref{XXcomm}), and $R$ is the network energy matrix computed in (\ref{Rjj}), (\ref{Rjk}). \hfill$\blacksquare$
\end{lem}

In view of the partitioning (\ref{FFF})  of the matrix $F$ and  the parameterization of the matrix $B$ in (\ref{AB}),  the condition (\ref{FB0}) is equivalent to
\begin{equation}
\label{FB01}
    \sum_{j \in V}
    F_j \Theta_j M_{kj}^\rT
    =
    F_k \Theta_k M_k^\rT
    +
    \sum_{j \in N_k^+}
    F_j \Theta_j N_{jk}^\rT D_{kj} = 0,
    \qquad
    k \in V,
\end{equation}
where the block-diagonal structure of the CCR matrix $\Theta$ in (\ref{XXcomm}) is used together  with the blocks (\ref{Mjk}) of the network-field coupling matrix $M$. The representation (\ref{FB01}) of (\ref{FB0}) shows that an appropriate ``tuning'' of the field-mediated coupling between the constituent OQHOs in the network can provide an additional freedom in choosing the matrix $F$ subject to (\ref{FB0}).

The absence of the quantum diffusion term (involving $W$) on the right-hand side of (\ref{phidot}) makes it an ODE and allows $\varphi$ to be interpreted as a vector of subnetwork variables,  partially isolated from the external fields. Compared to the completely autonomous subsystems arising in the quantum
Kalman decomposition \cite{ZGPG_2018}, this isolation is partial because $\varphi$ is indirectly affected by $W$ through ``complementary'' network variables which form a quantum process
\begin{equation}
\label{psiX}
    \psi(t) := T X(t).
\end{equation}
Here, $T \in \mR^{(n-s)\x n}$  is any matrix whose rows complement those of the matrix $F$ in Lemma~\ref{lem:iso} to a nonsingular matrix:
\begin{equation*}
\label{SFT}
    S :=
    \begin{bmatrix}
      F\\
      T
    \end{bmatrix}
    \in \mR^{n\x n},
    \qquad
        \det S\ne 0.
\end{equation*}
Therefore, an appropriate  similarity transformation allows the network QSDE (\ref{dX}) to be represented as an ODE for $\varphi$ and a QSDE for $\psi$ as
\begin{equation}
\label{phipsi}
    \mathop{\varphi}^{\centerdot} = a_{11}\varphi + a_{12}\psi,
    \qquad
    \rd \psi
    =
    (a_{21}\varphi + a_{22}\psi) \rd t + b \rd W,
\end{equation}
where the matrices $a_{11} \in \mR^{s\x s}$, $a_{12} \in \mR^{s\x (n-s)}$, $a_{21} \in \mR^{(n-s)\x s}$, $a_{22} \in \mR^{(n-s)\x (n-s)}$ and $b\in \mR^{(n-s)\x m}$ are given by
\begin{equation*}
\label{ab}
    a:=
    {\begin{bmatrix}
      a_{11} & a_{12}\\
      a_{21} & a_{22}
    \end{bmatrix}}
    =
            SAS^{-1},
    \qquad
    b := TB.
\end{equation*}
The decomposition (\ref{phipsi}) of the network dynamics can be viewed as an  interconnection of a partially isolated subnetwork $\Phi$ with the vector $\varphi$ of internal variables in (\ref{phiX}) interacting with a complementary subnetwork $\Psi$ with the vector $\psi$ of internal variables in (\ref{psiX}), with only $\Psi$ being directly affected by the external quantum Wiener process $W$ in (\ref{W}). In this sense,   $\Psi$ ``shields'' $\Phi$ from $W$;   see Fig.~\ref{fig:PhiPsi}.
\begin{figure}[htbp]
{\centering
\unitlength=0.75mm
\linethickness{0.4pt}
\begin{picture}(50.00,55)
    \put(10,10){\framebox(30,30)[cc]{}}
    \put(0,0){\framebox(50,50)[cc]{}}
    \put(15,28){\makebox(0,0)[cc]{$\varphi$}}
    \put(35,28){\makebox(0,0)[cc]{$\psi$}}
    \put(20,20){\framebox(10,10)[cc]{$\Phi$}}
    \put(45,25){\makebox(0,0)[cc]{$\Psi$}}
    \put(40,25){\vector(-1,0){10}}
    \put(20,25){\vector(-1,0){10}}
    \put(60,25){\vector(-1,0){10}}
    \put(63,25){\makebox(0,0)[lc]{$W$}}
\end{picture}
\caption{A schematic representation of (\ref{phipsi}) as an interconnection of quantum subnetworks $\Phi$, $\Psi$, where $\Phi$ is affected by the external fields $W$ only through $\Psi$ which interacts with $W$.}
\label{fig:PhiPsi}}
\end{figure}
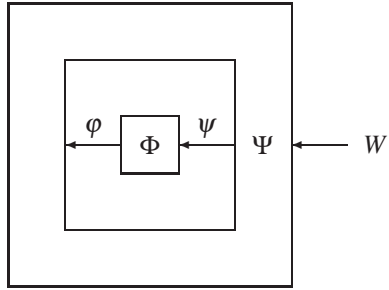

The special choice of the matrix $F$ in Lemma~\ref{lem:iso} subject to (\ref{FB0}) makes  the coefficient  $\dot{\Delta}(0)$ in (\ref{Deldot}) vanish. Hence, the quadratic term becomes dominant in (\ref{Delasy0}) which acquires the form \cite[Lemma~2]{VP_2024_arxiv}
\begin{equation}
\label{Del0asy}
  \Delta(t)
  =
    \frac{1}{2}\ddot{\Delta}(0)t^2
   + O(t^3),
   \qquad
   \ddot{\Delta}(0)
   =
   2\|G \sqrt{P}\|^2,
\end{equation}
as $t \to 0+$, with the matrices $G$ from (\ref{G}) and $P$ from (\ref{EXX0}). The fact that the short-horizon asymptotic behaviour (\ref{Del0asy}) is  quadratic (rather than linear) in time is a consequence of the partial isolation of the subnetwork $\Phi$ from the external fields (see Fig.~\ref{fig:PhiPsi}).   This is a qualitatively slower growth of the function $\Delta$ compared to the case of matrices $F$ satisfying (\ref{pos1}).
Accordingly, the partial subnetwork isolation prolongs the memory decoherence time in the high-fidelity limit as the following adaptation of \cite[Theorem 1]{VP_2024_arxiv} shows.

\begin{lem}
\label{lem:tauasy}
Suppose the conditions of Lemma~\ref{lem:iso} are satisfied together with (\ref{Fpos}) and
\begin{equation*}
\label{Gpos}
    G
    \sqrt{P}
    \ne 0.
\end{equation*}
Then the memory decoherence time (\ref{tau}) behaves asymptotically as
\begin{equation}
\label{tau0asy}
    \tau(\eps)
    \sim
    \frac{\|F\sqrt{P}\|}{\|G\sqrt{P}\|}
    \sqrt{\eps}
    =:
    \wh{\tau}(\eps),
    \qquad
    {\rm as}\
    \eps\to 0+.
\end{equation}
\hfill$\blacksquare$
\end{lem}

As a result of the partial subnetwork isolation,   the asymptotic behaviour (\ref{tau0asy}) in the high-fidelity limit  yields a longer decoherence time compared to the case of (\ref{pos1}), where $\tau(\eps)$ is asymptotically linear with respect to $\eps$ in view of (\ref{ratio}). Indeed, $\eps = o(\sqrt{\eps})$, as $\eps\to 0+$.  In the framework of (\ref{tau0asy}),  the maximization of $\tau(\eps)$ at a given small value of $\eps$ can also be replaced with its approximate  version (\ref{tauhatmax}). Assuming the matrices $F$  and $P$ to be fixed, the problem   (\ref{tauhatmax}),  with the approximate decoherence time  $\wh{\tau}(\eps)$ from (\ref{tau0asy}),  is equivalent to
the minimization of the denominator
\begin{equation}
\label{GP}
    \|G \sqrt{P}\|
    =
    2\|F\Theta R \sqrt{P}\|
\end{equation}
(see (\ref{G})),
which is a convex function of the network energy matrix $R$. The following theorem is a counterpart of Theorem~\ref{th:tauhatmax1} in the partial isolation setting. Accordingly, (\ref{pos1}) is replaced with (\ref{FB0}), and (\ref{tau12}) is replaced with (\ref{tau0asy}).

\begin{thm}
\label{th:Rjk}
Suppose the network-field and field-mediated coupling matrices $M_j$, $N_{jk}$ are fixed along with the matrices $D_{jk}$ and the individual energy matrices $R_j$ of the component OQHOs in the network. Also,
suppose
the fidelity level $\eps$ in (\ref{tau}) and the matrix $F$ in (\ref{phiX}), (\ref{FF})
are fixed together with the matrix
$P$  in (\ref{EXX0}), and the conditions (\ref{Fpos}), (\ref{FB0}) are satisfied. Then the direct energy coupling matrices $R_{jk}^0$  deliver a solution to the
problem (\ref{taumaxRjk})
of maximizing the approximate decoherence time $\wh{\tau}(\eps)$ from (\ref{tau0asy}) if
and only if they satisfy (\ref{Rjksol}) with
the maps $g_{jk}$ from (\ref{gjk}),  and the maps $K_{jk}$ from (\ref{KLjk}) being replaced with
\begin{equation}
\label{Kjk1}
    K_{jk}(R)
    :=
    2
    (\bS
    (\Theta \Sigma \Theta RP))_{jk}.
\end{equation}
\hfill$\square$
\end{thm}
\begin{proof}
Although (\ref{Kjk1}) can be obtained from (\ref{KLjk}) by noting that, in view of (\ref{FF}),  $\Sigma B = F^\rT FB = 0$ in the partial subnetwork isolation case (\ref{FB0})  and hence, $L_{jk}=0$ in (\ref{Ljk}), we will provide an independent proof of the present theorem. As mentioned above, the problem (\ref{taumaxRjk}),  with the approximate decoherence time  $\wh{\tau}(\eps)$ from (\ref{tau0asy}), reduces to the minimization of the quantity (\ref{GP}). The latter is equivalent to minimizing
\begin{equation}
\label{fRjk}
    \phi
    :=
    \frac{1}{2}\|F\Theta R \sqrt{P}\|^2
    =
    \frac{1}{16}
    \ddot{\Delta}(0)
\end{equation}
(cf. (\ref{G}),  (\ref{Del0asy}))
as a convex quadratic function of the     direct energy coupling matrices $R_{jk}^0=R_{kj}^{0\rT} \in \mR^{n_j\x n_k}$, with
$j \in V$,
    $k \in N_j^0$, and the $\frac{1}{2}$-factor  introduced for convenience.   Therefore, (\ref{fRjk}) achieves a global minimum if and only if the partial Frechet derivatives  vanish:
\begin{equation}
\label{phi'0}
    \d_{R_{jk}^0} \phi = 0,
    \qquad
    \
    j \in V,\
    k \in N_j^0.
\end{equation}
By using the variational identity $\delta (\|z\|^2) = 2\bra z, \delta z\ket $ which holds for a real matrix $z$, the first variation of $\phi$ in  (\ref{fRjk}) with respect to the matrix $R_{jk}^0$ can be computed, similarly to (\ref{deltaDelRjk}), as
\begin{equation}
\label{deltaf}
    \delta \phi
    =
    \bra
    F\Theta R \sqrt{P},
    F\Theta (\delta R) \sqrt{P}
    \ket
    =
    -
    \bra
    \Theta F^\rT F\Theta R P,
    \delta R
    \ket
    =
    -
    \bra
    \bS(\Theta \Sigma \Theta R P),
    \delta R
    \ket
    =
    -2
    \bra
    (\bS(\Theta \Sigma\Theta R P))_{jk},
    \delta R_{jk}^0
    \ket.
\end{equation}
Here, use is also made of  the antisymmetry of the CCR matrix $\Theta$ in (\ref{XXcomm}) and the symmetry of $P$ in (\ref{EXX0}) along with (\ref{FF}). The relation (\ref{deltaf}) leads to the partial Frechet derivative
\begin{equation}
\label{f'}
    -\d_{R_{jk}^0} \phi
    =
    2(\bS(\Theta \Sigma\Theta R P))_{jk}
    =
    g_{jk}(R_{jk}^0) + K_{jk}(\breve{R}_{jk})
\end{equation}
represented in terms of (\ref{gjk}), (\ref{Kjk1}).
By (\ref{phi'0}), (\ref{f'}),   the condition (\ref{Rjksol}), with the maps $K_{jk}$ from (\ref{Kjk1}), is necessary and sufficient for optimality in the problem (\ref{taumaxRjk}) with the approximate decoherence time from (\ref{tau0asy}) in the partial subnetwork isolation case (\ref{FB0}).
\end{proof}

Note that Theorem~\ref{th:Rjk} is essentially an adaptation of Theorem~\ref{th:tauhatmax1} to the partial subnetwork isolation setting, with both of them reducing to the minimization of $\ddot{\Delta}(0)$ from (\ref{Delddot}) which takes a specific form (\ref{Del0asy}) in this case. Accordingly, the previous remarks on the general structure and numerical solution of the equations (\ref{Rjksol}) for an optimal direct energy coupling remain  applicable in this case too. At the same time, in the partial subnetwork isolation case, the maps $K_{jk}$ in (\ref{Kjk1}) do not involve the network-field and field-mediated coupling matrices. Therefore, other than affecting the choice of the matrix $F$ subject to (\ref{FB0}),   these matrices enter the equations (\ref{Rjksol}) only through the matrices $\breve{R}_{jk}$.

\section{CONCLUSION}
\label{sec:conc}

For a network of OQHOs with direct energy and field-mediated couplings, driven by external quantum fields, we have considered the problem of maximizing the high-fidelity approximation of the memory decoherence time (associated with the weighted mean-square deviation of the network variables from their initial conditions)  over the direct energy coupling matrices. Using the convex quadratic structure of this problem, we have obtained a necessary and sufficient condition of  optimality in the form of a set of Sylvester-like linear algebraic equations for blocks of a matrix whose sparsity is specified by the edges of the direct energy coupling graph of the network.  The approximate decoherence time maximization problem and its solution have also been discussed for a partially isolated  subnetwork with a qualitatively longer decoherence time and thus with an  enhanced memory system performance. The development of an algorithm for numerical solution of the resulting sets of equations can benefit from a combination of vectorization techniques with advanced graph theoretic considerations. Other possible directions of research on this topic include the extension of the quantum memory optimization approach  beyond the high-fidelity (that is, short-horizon) approximations and to finite-level open quantum networks with multipartite interactions.

\end{document}